\documentclass[twoside, a4paper,10pt]{article}
\usepackage{amssymb}
\usepackage{IEEEtrantools}
\usepackage[mathscr]{eucal}
\usepackage[dvips]{graphicx}
\usepackage{amsmath}
\usepackage{amsthm}
\usepackage{pstricks}
\usepackage{caption}
\usepackage{esint}
\usepackage{cite}
\usepackage{cancel}
\usepackage{tikz,bm}
\usetikzlibrary{arrows}

\usetikzlibrary{arrows.meta}
\def \ni{\noindent}

\newcommand{\be}{\begin{equation}}
\newcommand{\ee}{\end{equation}}
\newcommand{\ben}{\begin{equation*}}
\newcommand{\een}{\end{equation*}}
\newcommand{\bes}{\begin{eqnarray}}
\newcommand{\ees}{\end{eqnarray}}
\newcommand{\besn}{\begin{IEEEeqnarray*}{rCl}}
\newcommand{\eesn}{\end{IEEEeqnarray*}}

\newcommand{\txt}{\textrm}

\newcommand{\tr}{\txt{Tr}}

\newtheorem*{theorem*}{Theorem}

\newtheorem*{definition*}{Definition}
\newtheorem*{lemma*}{Lemma}
\newtheorem*{prop*}{Proposition}
\newtheorem*{corollary*}{Corollary}

\DeclareFontFamily{U}{mathx}{}
\DeclareFontShape{U}{mathx}{m}{n}{<-> mathx10}{}
\DeclareSymbolFont{mathx}{U}{mathx}{m}{n}
\DeclareMathAccent{\widecheck}{0}{mathx}{"71}

\title{A Concentration of Measure Phenomenon in Lattice Yang-Mills}
\author{T. Tlas}
\date{}

\begin{document}

\maketitle
\thispagestyle{empty}

\begin{abstract}
\ni We demonstrate that the pushforward of the product of Haar measures by the lattice Yang-Mills action concentrates as a Gaussian. It is also sketched how using this fact, one can recover the strong-coupling expansion. \newline
\end{abstract}

The aim of this note is to explore the concentration of measure phenomenon \cite{ledoux} in the context of lattice Yang-Mills theory. We shall see below that while this phenomenon does indeed manifest in this setting, it does not allow us to go beyond what is already known by using other methods. This is due to the fact that, in a certain sense which will become clear below, measure concentration and action minimization work against each other. Nevertheless, the theorem demonstrated below, even though it does not lead to new results in the current setting, is sufficiently interesting and instructive to be presented as the same methods can be (and in fact are) more successful in other situations. \newline

Let us begin by describing our setup. We shall deal with the Euclidean, lattice Yang-Mills theory in $D$ dimensions with periodic boundary conditions. For concreteness, we work with the $U(N)$ group, but essentially everything below can be trivially adapted to the other groups of physical interest. The lattice sites are labelled by $x$. The positive edge directions are labelled by greek letters. The notation $x + \mu$ stands for the lattice site which is displaced from $x$ one step in the direction defined by $\mu$. The number of lattice sites is denoted by $K$. We attach to every edge at each site an element $U_{x, \mu} \in U(N)$, and are interested in performing the following integral

\besn
Z &  =  &  \int dU e^{- \beta S(U)} \\
& = &  \int \prod_{x, \mu} d U_{x, \mu} \exp \Bigg  [  \frac{2 N^2}{\lambda} \sum_{\{\mu, \nu\}} \frac{1}{N} \Re \Big ( \tr \big ( U^\dagger_{x, \nu} U^\dagger_{x  + \nu, \mu} U_{x +\mu, \nu} U_{x, \mu}  \big )  \Big ) \Bigg ] ,
\eesn

where the sum goes over all (unordered) pairs of $\mu$ and $\nu$ such that $\mu \neq \nu$, or equivalently, over all plaquettes of the lattice. The $dU$'s stand for the Haar measures and $\lambda$ is the 't Hooft coupling. \newline

We want to find the asymptotic of the expression above in the limit $N \to \infty$ while holding $\lambda$ fixed. We shall do so by pushing forward the product of Haar measures to $\mathbb{R}$ via the function 

\be
\label{eq:action}
t(U_{x, \mu}) = \frac{1}{K} \sum_{\{\mu, \nu\}}  \frac{1}{N} \Re \Big ( \tr \big ( U^\dagger_{x, \nu} U^\dagger_{x  + \nu, \mu} U_{x +\mu, \nu} U_{x, \mu}  \big )  \Big ).
\ee

We thus get that 

\be
\label{eq:partition}
Z = \int e^{\frac{2 N^2}{\lambda} K t} d \rho_N[t],
\ee

where $\rho_N$ is the pushforward measure under the function above. What can we say about this measure? The key fact is furnished by the following

\begin{theorem*}
As $N \to \infty$ the rescaled measures $\rho_N(N t)$ converge weakly to the measure $ \sqrt{\frac{2 K}{\pi D (D-1)}} e^{-  \frac{ 2 K}{D (D-1) } t^2} dt$.
\end{theorem*}

\begin{proof}
Let $m_l$ stand for the $l$-th moment of Gaussian above. In other words, let

\ben
m_l = 
\begin{cases}
\quad \quad 0 & \text{if} \quad  l \quad  \text{is odd}\\
 \frac{(l-1)!!}{  \big (  \frac{4 K }{D (D-1)}  \big )^{\frac{l}{2}} }  & \text{if}  \quad l  \quad \text{is even}.
\end{cases}
\een

It will be shown below that the $l$-th moment of $\rho_N(t)$ is equal to $\frac{m_l}{N^l} + O(\frac{1}{N^{l+1}})$, which of course implies that the $l$-th moment of $\rho_N(Nt)$ is $m_l + O(\frac{1}{N}) \to m_l$ as $N \to \infty$. Therefore, combining the method of moments together with the fact that the moment problem for a Gaussian is determined, the theorem will be proven. \newline

We now adapt the graphical representation common in the literature \cite{creutz} to our goals. A line segment decorated by an arrow will stand for $U_{ij}$ and an adjacent line segment with the arrow pointing in the other direction will denote $U^\dagger_{i'j'}$. Since all the expressions below will be symmetric between $U$ and $U^\dagger$, we won't need to specify which is which. An undecorated line segment will stand for a Kronecker $\delta_{ij}$. Attaching line segments to each other corresponds to multiplying the corresponding matrices in the order dictated by thinking of them as being parallel transports. See Figure \ref{fig1} for a summary of this notation. Note that we do not require all these segments to be straight or horizontal, as the only fact that matters is how they are connected to each other. \newline

\begin{figure}[h]
\begin{tikzpicture}
\draw[-{Latex}] (0, 5) -- (1, 5);
\draw (0.9,5) -- (2, 5);
\node at (1, 5.5){$U_{ij}  \, \, \txt{or} \, \, U^\dagger_{ij}$};
\draw[-{Latex[reversed]}   ](3, 5) -- (4,5);
\draw (4,5) -- (5,5);
\draw[-{Latex}   ](3, 4.5) -- (4,4.5);
\draw (3.9,4.5) -- (5,4.5);
\node at (4, 5.5) {$U_{ij} U^\dagger_{i'j'}$};
\draw (6, 5) -- (8, 5);
\node at (7, 5.5) {$\delta_{ij}$};
\draw[-Latex] (9, 5) -- (10, 5);
\draw(10, 5) -- (11, 5);
\draw[ -{Latex[open]}](11,5) -- (11, 6);
\draw (11, 5.97) -- (11, 7);
\node at (10, 4.5) {$U_{jk}$};
\node at (11.5, 6) {$U'_{ij}$};
\node at (9.7, 6) {$ \sum_j U'_{ij} U_{jk}   $     };

\end{tikzpicture}
\caption{\textit{Summary of a graphical notation used below. Note the order of the product in the sum in the rightmost expression. It is the one consistent with treating $U$ and $U'$ as being parallel transports (acting conventionally to the right).}}
\label{fig1}
\end{figure}
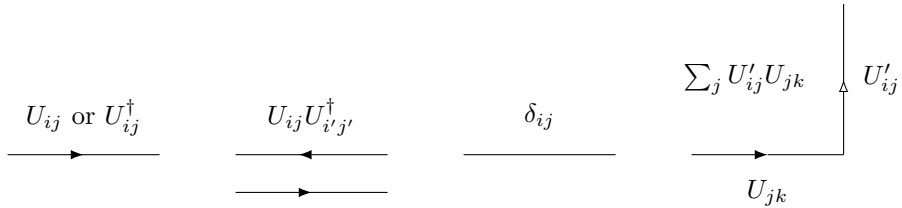

We shall deal with integrals of polynomials in the matrix entries (as well as their conjugates) with respect to the Haar measure. Such integrals are well-studied with their evaluations being known explicitly for any $N$ (see \cite{weingarten1} for a lovely overview). However, we won't have an occasion below to use anything beyond the lowest asymptotic order obtained in \cite{weingarten2}. The formula we need is

\bes
\label{eqnarray:integral}
\int dU U_{i_1 j_1} \dots U_{i_n j_n} U^*_{i'_1 j'_1} \dots U^*_{i'_m j'_m} & = & \frac{\delta_{nm}}{N^n} \sum_\sigma \delta_{i_1 i'_{\sigma(1)}} \delta _{j_1 j'_{\sigma(1)}} \dots \delta_{i_n i'_{\sigma(n)}} \delta _{j_n j'_{\sigma(n)}} \nonumber \\
& + & O(\frac{1}{N^{n+1}}),
\ees

where the sum goes over all permutations of $n$ elements. A graphical representation of the formula above for the cases $n=2$ and $n=4$ is shown in Figure \ref{fig2}. \newline

\begin{figure}[h]
\begin{tikzpicture}[scale = 0.89]

\node at (0.3,1) { $ \int dU$};
\draw[-{Latex}] (1,0) -- (1, 1);
\draw (1,0.9) -- (1,2);
\draw[-{Latex[reversed]}] (1.5,0) -- (1.5, 1);
\draw (1.5,0.9) -- (1.5,2);
\node at (2, 1) {$=$}; 
\node at (2.5, 1) {  $\frac{1}{N}$};

\draw (2.8,0) -- (2.8,0.7);
\draw (2.8, 0.7) -- (3.3, 0.7);
\draw (3.3, 0) -- (3.3, 0.7);

\draw (2.8,1.3) -- (2.8,2);
\draw (2.8, 1.3) -- (3.3, 1.3);
\draw (3.3, 1.3) -- (3.3, 2);

\node at (3.6, 1) {,};

\node at (4.3,1) { $\int dU$};
\draw[-{Latex}] (5,0) -- (5, 1);
\draw (5,0.9) -- (5,2);
\draw[-{Latex[reversed]}] (5.5,0) -- (5.5, 1);
\draw (5.5,0.9) -- (5.5,2);
\draw[-{Latex}] (6,0) -- (6, 1);
\draw (6,0.9) -- (6,2);
\draw[-{Latex[reversed]}] (6.5,0) -- (6.5, 1);
\draw (6.5,0.9) -- (6.5,2);

\node at (6.9, 1) {$=$};

\node at (7.4, 1) { $ \frac{1}{N^2} $ };
\node at (7.9, 1) {  $ \Bigg \{ $ };

\draw (8.2,0) -- (8.2,0.7);
\draw (8.2, 0.7) -- (8.7, 0.7);
\draw (8.7, 0) -- (8.7, 0.7);

\draw (8.2,1.3) -- (8.2,2);
\draw (8.2, 1.3) -- (8.7, 1.3);
\draw (8.7, 1.3) -- (8.7, 2);

\draw (8.9,0) -- (8.9,0.7);
\draw (8.9, 0.7) -- (9.4, 0.7);
\draw (9.4, 0) -- (9.4, 0.7);

\draw (8.9,1.3) -- (8.9,2);
\draw (8.9, 1.3) -- (9.4, 1.3);
\draw (9.4, 1.3) -- (9.4, 2);

\node at (9.6, 1) {$+$};

\draw  (9.8, 0) -- (9.8 , 0.9);
\draw (9.8, 0.9) -- (11. 2, 0.9);
\draw (11.2, 0) -- (11.2, 0.9);

\draw (10.2, 0) -- (10.2, 0.7);
\draw (10.2, 0.7) -- (10.8, 0.7 );
\draw (10.8, 0) -- (10.8, 0.7);

\draw  (9.8, 1.1) -- (9.8 , 2);
\draw (9.8, 1.1) -- (11. 2, 1.1);
\draw (11.2, 1.1) -- (11.2, 2);

\draw (10.2, 1.3) -- (10.2, 2);
\draw (10.2, 1.3) -- (10.8, 1.3 );
\draw (10.8, 1.3) -- (10.8, 2);

\node at (11.5, 1) { $ \Bigg \}$ };

\node at (12.6, 1 ) { $ + \, O (\frac{1}{N^3}) $};

\end{tikzpicture}

\caption{\textit{A pictorial representation of (\ref{eqnarray:integral}) in the cases $n=2$ and $n=4$. Graphically, integration over a bundle of edges amounts to cutting them into half-edges, and then linking or coupling the half edges in pairs on one side of the cut with each other, with matching coupling of the half-edges on the other side of the cut as well. Incidentally, note that in the first case, the dominant asymptotic happens to be the exact answer. There are no further corrections.}}
\label{fig2}
\end{figure}

We are now ready to compute the moments of $\rho_N$. From the definition of the $\rho_N$ we have that

\be
\label{eq:moment}
\int t^l d\rho_N(t) = \int t^l(U_{x, \mu} ) \prod_{x, \mu} dU_{x, \mu}
\ee

From (\ref{eq:action}), it is obvious that the expression above would be a sum of integrals of products of $l$ plaquettes. We thus concentrate on computing one such term. From (\ref{eqnarray:integral}), it is clear that the result is nonvanishing if every edge appears an even number of times, half of them carrying $U$ and the other half carrying $U^\dagger$. Let us now find the leading asymptotic of each such term. \newline

First, from (\ref{eq:action}), note that each plaquette comes with a factor of $\frac{1}{N}$. Thus, we have an overall factor of $\frac{1}{N^l}$. Second, from (\ref{eqnarray:integral}), we have that the integral contributes a further factor of $\frac{1}{N}$ for each pair of edges integrated over. Since there are $4l$ edges in total, we get a further factor of $\frac{1}{N^{2l}}$. Finally, again from (\ref{eqnarray:integral}), we can see that integration is going to produce products of delta functions. Since there are no free indices in the expression integrated, we will simply have a collection of traces of identity. Each of them is going to contribute a factor of $N$. It remains to count how many such traces will be produced. \newline

To this end, pick a pair of edges which have been cut with the half-edges coupled after integration, and follow one of the two pairs of the coupled half-edges. After reaching a corner, either these two edges will be coupled again, thus producing a loop, or they will coupled to different half-edges. Note that in the former case, we get a loop out of 4 half-edges, while in the latter, we will use up more than 4 half-edges to form a loop. See Figure \ref{fig3} for a pictorial clarification. \newline

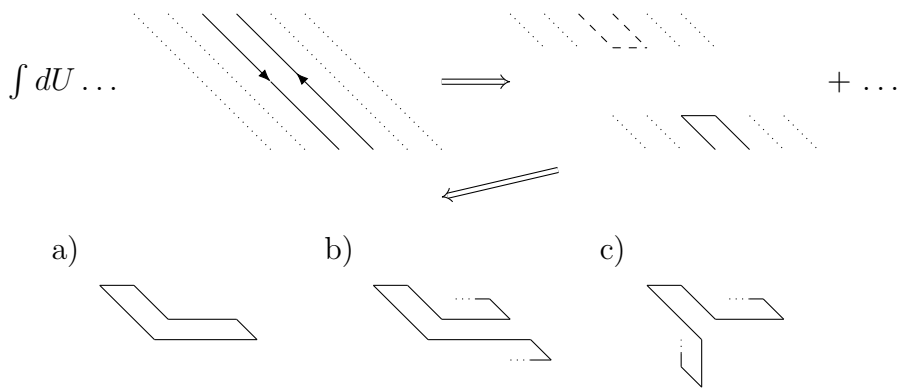
\begin{figure}[h]
\begin{tikzpicture}[scale = 0.9]

\node at (0,6) {\large $ \int dU \dots$};
\draw[dotted] (1,7) -- (3, 5);
\draw[dotted] (1.5,7) -- (3.5, 5);
\draw[-{Latex}    ] (2, 7) -- (3, 6);
\draw (3,6) -- (4, 5);
\draw[-{Latex[reversed]}    ] (2.5, 7) -- (3.5, 6);
\draw (3.5,6) -- (4.5, 5);
\draw[dotted] (3, 7) -- (5, 5);
\draw[dotted] (3.5, 7) -- (5.5, 5);

\draw[-implies,  double equal sign distance] (5.5,6) -- (6.5,6);

\draw[dotted] (6.5, 7) -- (7, 6.5);
\draw[dotted] (8, 5.5    )  -- (8.5, 5);

\draw[dotted] (7, 7) -- (7.5, 6.5 );
\draw[dotted] (8.5  , 5.5 ) -- (9, 5);

\draw[dashed] (7.5, 7) -- (8, 6.5);
\draw[dashed] (8, 7) -- ( 8.5, 6.5);
\draw[dashed] (8, 6.5) -- (8.5, 6.5);

\draw   ( 9.5  ,  5.5   )  --  ( 10 , 5 ) ;
\draw   ( 9, 5.5 ) -- ( 9.5 , 5);
\draw (9.5 , 5.5) -- ( 9, 5.5);

\draw[dotted] (8.5, 7) -- (9, 6.5);
\draw[dotted] (10 , 5.5 ) -- (10.5, 5);

\draw[dotted] (9, 7) -- (9.5,  6.5);
\draw[dotted] (10.5  , 5.5 ) -- (11, 5);

\node at (11.7, 6) { \large $ + \, \dots $};

\draw (0.5, 3) -- (1.3, 2.2);
\draw (0.5, 3) -- (1, 3);
\draw (1, 3) -- (1.5, 2.5);
\draw (1.5, 2.5) -- (2.5, 2.5);
\draw (2.5, 2.5) -- (2.8 , 2.2);
\draw (1.3, 2.2) -- (2.8, 2.2);

\draw (4.5, 3) -- (5.3, 2.2);
\draw (4.5, 3) -- (5, 3);
\draw (5, 3) -- (5.5, 2.5);
\draw (5.5, 2.5) -- (6.5, 2.5);

\draw (6.5, 2.5) -- ( 6.2, 2.8);
\draw (6.2, 2.8) -- (6, 2.8);
\draw[dotted] (6, 2.8) -- ( 5.7, 2.8);

\draw (5.3, 2.2) -- (6.8, 2.2);

\draw (6.8, 2.2) -- (7.1, 1.9);
\draw (7.1, 1.9) -- (6.8, 1.9);
\draw[dotted] (6.8, 1.9) -- (6.5, 1.9);

\draw (8.5, 3) -- (9.3, 2.2);
\draw (8.5, 3) -- (9, 3);
\draw (9, 3) -- (9.5, 2.5);
\draw (9.5, 2.5) -- (10.5, 2.5);

\draw (10.5, 2.5) -- ( 10.2, 2.8);
\draw (10.2, 2.8) -- (10, 2.8);
\draw[dotted] (10, 2.8) -- ( 9.7, 2.8);

\draw (9.3, 2.2) -- (9.3, 1.5);
\draw (9.3, 1.5) -- ( 9, 1.8) ;
\draw (9, 1.8) -- (9, 2);
\draw[dotted] (9, 2) -- (9, 2.2);

\draw[-implies,  double equal sign distance] (7.2,4.7) -- (5.5,4.3);

\node at (  0, 3.5   ) {\large a)};
\node at (  4, 3.5   ) {\large b)};
\node at (  8, 3.5   ) {\large c)};

\end{tikzpicture}

\caption{\textit{A pictorial representation of the argument given in the main text. We perform the integral over the group elements living on the edges. Only one of these integrals is explicitly written out with the rest being represented by the $\dots$ under the integral sign. As represented in Figure \ref{fig2}, after performing the integral over the group, the edges carrying that group element and its inverse are cut into half-edges with these halves linked up in pairs in identical ways on the two sides of the cut. In the figure, we focus on two edges represented by solid lines with arrow whose half-edges will be coupled. We focus on only one coupled pair, denoted by a solid line. The dashed one stands for the other one. The dotted lines stand for the other copies of the same edge carrying $U$ and $U^\dagger$ which, after integration, are of course coupled among themselves as well, which is indicated by the break in the lines. The $+ \, \dots$ stands for all the other terms in the sum in (\ref{eqnarray:integral}) which correspond to a different coupling of the half-edges. Also, note that since our concern here is the factors of $N$ obtained by following the coupling of edges, we haven't indicated either the factors of $\frac{1}{N}$ coming from the definition of $t(U)$, or those on the right hand side of (\ref{eqnarray:integral}). Now, following the two singled out coupled half-edges, we will have one of three cases in the lower half of the figure. In case (a), they meet at the corner with two coupled half-edges forming a loop. Note that  this loop is constructed out of four half-edges. Also note that the two plaquettes from which these four half-edges came must actually be the same plaquette, as two of their edges coincide. In cases (b) and (c), the original two half-edges are not coupled back immediately forming a loop, but instead, are coupled to other half edges. The difference between these two cases is that in (b), the two plaquettes from which these half-edges come coincide as in (a), while in (c) they don't. } }
\label{fig3}
\end{figure}

Since in total we have $8l$ half-edges, and since at least 4 are needed to form a single loop, the total number of loops is no more than $2l$. Thus, the total factor that they contribute is no more than $N^{2l}$. \newline

Putting everything together, we get that the dominant contribution asymptotically to our expression is of order $\frac{1}{N^l}$. This comes from those terms which, when represented graphically, contain only case (a) from Figure \ref{fig3}. How can the latter situation arise? It should be clear that it would only happen if we can group the plaquettes into pairs, with one copy in the pair oriented one way, while the other the opposite way, and then couple the half-edges as indicated in Figure \ref{fig4} because any other scenario would produce cases (b) or (c). \newline

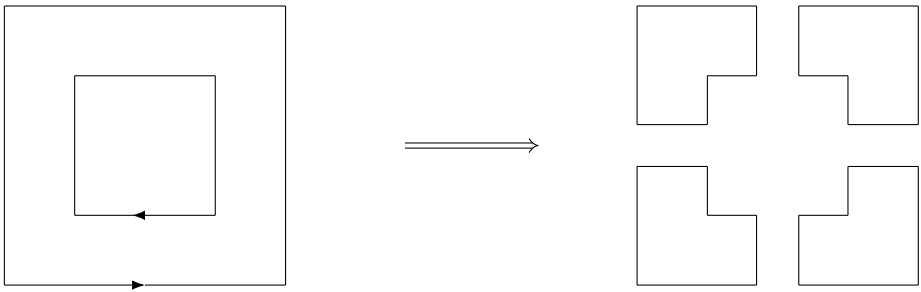
\begin{figure}[h]
\begin{tikzpicture}[scale = 0.925]

\draw[-{Latex}] (0, 0) -- (2, 0);
\draw (2,0) -- (4,0);
\draw (4, 0) -- (4, 4);
\draw (4, 4) -- (0, 4);
\draw (0, 0) -- (0, 4);

\draw[-{Latex[reversed]}] (1, 1) -- (2, 1);
\draw (2,1) -- (3,1);
\draw (3, 1) -- (3, 3);
\draw (3, 3) -- (1, 3);
\draw (1, 1) -- (1, 3);

\draw[-implies,  double equal sign distance] (5.7,2) -- (7.6,2);

\draw (9, 0) -- (10.7, 0);
\draw (11.3, 0) -- (13, 0);

\draw (10.7, 0) -- (10.7, 1);
\draw (11.3, 0) -- (11.3, 1);

\draw (13, 0) -- (13, 1.7);
\draw (13, 2.3) -- (13, 4);
\draw (12, 1.7) -- (13, 1.7);
 \draw (12, 2.3) -- (13, 2.3);

\draw (10, 1) -- (10.7, 1);
\draw (11.3, 1) -- (12, 1);

\draw (12, 1) -- (12, 1.7);
\draw (12, 2.3) -- (12, 3);

\draw (13, 4) -- (11.3, 4);
\draw (9, 4) -- (10.7, 4);

\draw (9, 0) -- (9, 1.7);
\draw (9, 2.3) -- (9, 4);

\draw (12, 3) -- (11.3, 3);
\draw (10.7, 3) -- (10, 3);

\draw (10, 1) -- (10, 1.7);
\draw (10, 2.3) -- (10, 3);

\draw (10, 1.7) -- (9, 1.7);
\draw (10, 2.3) -- (9, 2.3);
\draw (10.7, 3) -- (10.7, 4);
\draw (11.3, 3) -- (11.3, 4);

\end{tikzpicture}
\caption{\textit{This is a graphical representation of the only way two plaquettes can be coupled together in order to solely produce instances of case (a) from Figure \ref{fig3}. To see this, suppose the half-edges are coupled to produce case (a) for e.g., the bottom left corner. It follows that the other two half-edges coming from the bottom edge must be coupled together as well. Following them around the bottom right corner, we conclude that the corresponding two half-edges coming from the right edge must also be coupled together, for otherwise, we would have either case (b) or (c) at the bottom right corner. Proceeding like this around the plaquette, we conclude that the two half-edges coming from these two plaquettes must only couple among themselves, giving the picture on the right.  }}
\label{fig4}
\end{figure}

Since all such terms evaluate to the same value $\frac{1}{N^l}$, it remains to find the total number of such terms. In other words, we must determine the total number of ways we can group $l$ plaquettes into oppositely oriented pairs. For an odd $l$, this is impossible, and thus we do indeed have that the $l$-th moment of $\rho_N(t)$ is $O ( \frac{1}{N^{l+1}}) = \frac{m_l}{N_l}  +O ( \frac{1}{N^{l+1}}) $. \newline

For an even $l$, the number of such terms is equal to

\ben
(l-1)!!  \times 2 ^{\frac{l}{2}}  \times K^{\frac{l}{2}} \times  {D \choose 2} ^{\frac{l}{2}},
\een 

where the first factor comes from the number of ways one can subdivide a set of $l$ elements into pairs. The second is a consequence of making a choice of orientation for each plaquette in each pair (each plaquette enters with both orientations as is clear from (\ref{eq:action})). The third is the number of sites where the two paired plaquettes can be located. The fourth is the choice of the coordinate plane in space to which the palquettes are parallel. \newline

Recalling now that $t$ has a factor of $K$ in the denominator as well as a factor of $\frac{1}{2}$ (coming from the real part), we get that for even $l$, (\ref{eq:moment}) is equal to

\ben
\frac{1}{(2K)^l}  \times (l-1)!! \bigg [  2 K {D \choose 2}   \bigg ] ^{ \frac{l}{2}  } \times \frac{1}{N^l} + O(\frac{1}{N^{l+1}}) =   \frac{1}{N^l}   \frac{  (l-1)!! }{  \Big (  \frac{2K }{ {D \choose 2}  } \Big )^{\frac{l}{2}} } + O( \frac{1}{N^{l+1}}), 
\een

which is exactly what we want, thus completing the proof.
\end{proof}

The theorem above states essentially that  

\be
\label{eq:asymptotic}
d \rho_N[t] \sim \sqrt{\frac{2 KN^2}{\pi D (D-1)}}   e^{- \frac{2KN^2}{D(D-1)} t^2   } dt \qquad \txt{as} \qquad N \to \infty.
\ee

Let us thus replace $d \rho_N[t]$ in (\ref{eq:partition}) with the asymptotic Gaussian. Then, we get the following expression

\be
\label{eq:partition2}
\sqrt{\frac{2KN^2}{\pi D (D-1)}}  \int e^{ \frac{2 N^2 K}{\lambda} t   }  e^{ - \frac{2KN^2}{D(D-1)} t^2    } dt  \sim     e^{   \frac{ K N^2 D (D-1)  }{ 2 \lambda^2  }   }  .
\ee

How reliable is this result as an asymptotic of the partition function? In other words, how reliable is the replacement (\ref{eq:asymptotic}) in (\ref{eq:partition})? Comparing (\ref{eq:partition2}) to the known exact asymptotic for $D = 2$ obtained in \cite{grosswitten}, we see that the answer above coincides with the correct expression for $\lambda$ sufficiently large ($\lambda \geq 2$). But, it does not match with the correct behaviour for $0< \lambda \leq 2$; notably it does not detect the present of the phase transition. \newline

The reason for this is not hard to see. While the theorem above justifies the replacement (\ref{eq:asymptotic}) in an integral of the form $\int f(t) d\rho_N[t]$, in our case, the integrated function not only depends on $N$, it depends on it exponentially and with the same weight in $N$. Put differently, there are two processes which compete to decide the asymptotic of (\ref{eq:partition}). One of them is the concentration of measure phenomenon which is manifested by the Gaussian, the second term in the integrand in (\ref{eq:partition2}), trying to set $t = 0$. The other process is the minimization of the action contained in the first term in the integrand in (\ref{eq:partition2}), which is aiming to make the action as small as possible. This corresponds to pushing $t$ to its maximum value $t = {D \choose 2}$.\footnote{It is not hard to see from (\ref{eq:action}) that the maximum value of $t$, attained when e.g. all the $U$'s are equal to identity, is ${D \choose 2}$.} Which of these processes wins is decided by their relative weights in the exponents, i.e. by $\lambda$. It is thus clear that it is precisely when $\lambda$ grows, that the action term becomes less relevant than the measure one and vice versa. We can thus expect that the asymptotic replacement (\ref{eq:asymptotic}) is reliable in the limit $\lambda \to \infty$, which is of course the strong-coupling limit, and is unreliable in the limit $\lambda \to 0$. Alas it is exactly the latter limit which is the physically relevant one. Nonetheless, as is often the case with asymptotic analysis, the approximations often work much better than expected. This is exemplified above in the $D = 2$ case, where $\lambda \geq 2$ is already ``large'' and the asymptotic agrees precisely with the exact answer.\newline

Let us finish with three remarks. First, note that in view of the discussion above, the asymptotic replacement (\ref{eq:asymptotic}) can be thought of as an alternative derivation of the lowest order contribution to the strong coupling expansion. Now, how can we generate the higher orders? To do so, we should note that strictly speaking, (\ref{eq:asymptotic}) does not follow from the theorem. Rather, one can only conclude that 

\ben
d \rho_N[t] \sim \alpha_N \exp \bigg [ - N^2\Big ( c_2(N) t^2 + c_3(N) t^3 + c_4(N) t^4 + \dots        \Big ) \bigg ] dt
\een

where $c_2, c_3, c_4, \dots$ can each be expanded in powers of $\frac{1}{N}$, and $\alpha_N$ is chosen to normalize the expression on the right hand side to 1. To determine these higher order corrections, one would need to consider the next orders of corrections to the moments calculated above, and equate them with the asymptotic expansion of the moments of the expression on the right above. This can be done systematically, as at any given order, only finitely many new coefficients would be relevant; However, the graphical calculations rapidly get quite unwieldy. Since there is already a well-developed and understood strong coupling expansion, there is little incentive to develop these calculations in this setting.  \newline

The second remark we would like to make concerns the other limit $\lambda \to 0$. Can one use (\ref{eq:partition}) to calculate the asymptotic to $Z$ in that limit? Alas, as will be clear below, it does not look possible to go beyond the lowest order approximation. However, the argument and its conclusion are quite simple and instructive, so they are worth presenting. \newline

Note that if $\lambda \to 0$, we expect the ``action'' term to be the dominant one. Thus, we expect the main contribution to come from the neighbourhood of the point that minimizes the action. After taking the pushforward, this corresponds to the maximum possible value of $t$ which, as mentioned above, is ${D \choose 2}$. How does $d\rho_N[t]$ behave in the neighbourhood of this point? The exact answer would require finding the Hausdorff measure induced by the Haar measure on the level sets of the function defined in (\ref{eq:action}). This is a difficult undertaking which amounts to essentially obtaining exactly the partition function as a function $N, \lambda$ and $K$. However, if one writes the group elements $U$ as $e^X$, where $X$ is in the Lie algebra $su(N)$, and keeps only the lowest order terms in the expansion of the exponential, the problem simplifies considerably. This is because the equation $t(U) = \txt{constant}$ becomes a homogeneous one of order 2 if one only keeps the terms up to the quadratic order.\footnote{Strictly speaking, in our setting, the lowest nonvanishing order after the constant term is not the quadratic but the linear one. However, the linear contribution does not scale with $N$ while the quadratic one does, and thus, is the dominant one. Alternatively, one can recall the standard fact that the large $N$ limits of $U(N)$ and $SU(N)$ theories coincide and revert to the $SU(N)$ case if desired. } Now, if we count the degrees of freedom, taking into account gauge invariance, we see that $d \rho_N[t]$ should scale as $t^{ \frac{ (D-1) (N^2 - 1) K    }{2}     }$ in the neighbourhood of $t = {D \choose 2}$. Therefore, let us replace $d \rho_N[t]$ with this expression in (\ref{eq:partition}). We get

\besn
\int e^{ \frac{  2 N^2 K t  }{ \lambda }      } t^{ \frac{ (D-1) (N^2 - 1) K    }{2}     } dt  \sim  \int e^{N^2 K ( \frac{ 2 t  }{ \lambda }       +  \frac{(D-1)}{2} \ln t )  }    dt .
\eesn

The exponent is increasing as a function of $t$. Thus, the dominant contribution comes from the upper endpoint of the support of the integral which is ${D \choose 2}$. Therefore, the asymptotic free energy is given by $ N^2 K \Big ( \frac{ 2 {D \choose 2}  }{ \lambda   }   + \frac{D-1}{2} \ln {D \choose 2}   \Big )$. Comparing this with the expression obtained in \cite{grosswitten}, we see that we do recover the lowest order contribution to the expression there (which is $\frac{2N^2 K}{\lambda}$). Unfortunately, the method above does not seem to be adapted to go beyond the lowest order. \newline

The final remark we would like to make is to note that the reason that our results are only reliable in the large coupling limit is because, figuratively speaking, the measure and the action work against each other. One may wonder if there are other settings where they won't be in such a blatant opposition, and may actually work together to make the asymptotic more reliable in the continuum limit instead of less. It turns out that this is indeed the case for the principal chiral model, as is reported elsewhere \cite{pcm}. \newline

\textbf{Acknowledgments:} The author would like to thank J. Merhej for reading a preliminary version of this paper and for the numerous comments which have greatly improved its readability.

\texttt{{\footnotesize Department of Mathematics, American University of Beirut, Beirut, Lebanon.}
}\\ \texttt{\footnotesize{Email address}} : \textbf{\footnotesize{tamer.tlas@aub.edu.lb}}

\end{document}